\documentclass{article}
\usepackage[latin1]{inputenc}

\usepackage[english]{babel}
\usepackage{amsfonts}
\usepackage{amsmath,amssymb, amsthm}
\usepackage{enumerate}
\usepackage{verbatim}
\usepackage{caption}
\usepackage{tikz}
\usepackage{color}

\newtheorem{Claim}{Claim}
\newtheorem{Lemma}{Lemma}
\newtheorem{Corollary}{Corollary}
\newtheorem{Problem}{Question}
\newtheorem{Theorem}{Theorem}
\newtheorem{Proposition}{Proposition}

\title{Avoidability of long $k$-abelian repetitions}
\author {M. Rao, M. Rosenfeld\\LIP, ENS de Lyon, CNRS, Universit\'e de Lyon}

\begin{document}

\maketitle

\begin{abstract}
We study the avoidability of long $k$-abelian-squares and $k$-abelian-cubes on binary and ternary alphabets.
For $k=1$, these are M\"akel\"a's questions.
We show that one cannot avoid abelian-cubes of abelian period at least $2$ in infinite binary words, and therefore answering negatively one question from  M\"akel\"a.
Then we show that one can avoid $3$-abelian-squares of period at least $3$ in infinite binary words and $2$-abelian-squares of period at least 2 in infinite ternary words. Finally we study the minimum number of distinct $k$-abelian-squares that must appear in an infinite binary word.
\end{abstract}

\section{Introduction}
Avoidability of structures and patterns has been extensively studied in theoretical computer science since the work of Thue on avoidability of repetitions in words \cite{Thue1}. 
Thue showed that there are infinitely long ternary words avoiding squares (factors of the form $ww$ where $w$ is a word) and infinitely long binary words avoiding cubes (factors of the form $www$ where $w$ is a word). 

The avoidability of abelian repetitions has been studied since a question from Erd\H{o}s in 1957 \cite{erdos1, erdos2}.
A factor $uv$ is an abelian-square if $u$ is a permutation of the letters of $v$. 
Erd\H{o}s asked whether it is possible to avoid abelian-squares in an infinite word over an alphabet of size 4. 
(Abelian-square-free ternary words have a length of at most 7.)
After some intermediary results (alphabet of size 25 by Evdokimov \cite{Evdokimov1} and size 5 by Pleasant \cite{pleasant1}), Ker{\"a}nen answered positively Erd\H{o}s's question by giving a 85-uniform morphism (found with the assistance of a computer) whose fixed-point is abelian-square-free~\cite{keranen1}. 
Moreover, Dekking showed that it is possible to avoid abelian-cubes on a ternary alphabet and abelian-4th-powers over a binary alphabet \cite{Ternarycube}. 

Erd\H{o}s also asked if it is possible to avoid arbitrarily long ordinary squares on binary words. This question was answered positively by Entringer, Jackson and Schatz~\cite{Entringer1974159}. 
M\"akel\"a asked the following two questions about the avoidability of long abelian-cubes (resp. squares) on a binary (resp. ternary) alphabet:
\begin{Problem}[M\"akel\"a (see \cite{keranen2})]
\label{makquest}
Can you avoid abelian-cubes of the form $uvw$ where $|u|\geq 2$, over two letters ? - You can do this at least for words of length 250.
\end{Problem}
\begin{Problem}[M\"akel\"a (see \cite{keranen2})]
\label{makquestsquares}
Can you avoid abelian-squares of the form $uv$ where $|u|\geq 2$ over three letters ?  - Computer experiments show that you can avoid these patterns at least in words of length 450.
\end{Problem}

The notion of $k$-abelian repetition has been introduced recently by Karhum\"aki \emph{et al.} as a generalization of both repetition and abelian repetition \cite{Karhumaki1}.
One can avoid $3$-abelian-squares (resp. $2$-abelian-cubes) on ternary (resp. binary) words \cite{abeliancube}.
Following Erd\H{o}s's and M\"akel\"a's questions, one can ask whether it is possible to avoid long $k$-abelian-powers on binary (resp. ternary) words.  

In Section \ref{sec:mak}, we answer negatively M\"akel\"a's Question \ref{makquest}, and we propose a new version of the conjecture.
In Section \ref{sec:kab}, we show that one can avoid $3$-abelian-squares of period at least $3$ in binary words and $2$-abelian-squares of period at least $2$ in ternary words.
In Section \ref{sec:howmanykab}, we study the minimum number of distinct $k$-abelian-squares that must appear in an infinite binary word.
Finally, in Section \ref{sec:concl}, we explain the computer searches we use to find the morphisms of Section \ref{sec:kab} and Section \ref{sec:howmanykab}.

\section{Preliminaries and definitions}
We use terminology and notations of Lothaire~\cite{Lothaire}.
Let $\Sigma$ be a finite alphabet.
For a word $u\in \Sigma^*$ and $a\in \Sigma$, we denote by $|u|_a = |\{i : u[i]=a\}|$  the number of occurrences of the letter $a$  in $u$. 
For $u,w\in \Sigma^*$, we denote by $|u|_w = |\{i : u[i:i+|w|-1]=w\}|$ the number of occurrences of the factor $w$ in $u$.

Two words $u$ and $v$ are said to be \emph{abelian equivalent}, denoted $u \approx_a v$, if for every $a \in \Sigma$, $ |u|_a=|v|_a$, 
and they are said \emph{$k$-abelian equivalent} (for $k\ge 1$), denoted $u \approx_{a,k} v$, if for every $w\in \Sigma^*$ such that $|w| \leq k$, $|u|_w=|v|_w$. 
A word $u_1u_2\ldots u_n$ is a \emph{$k$-abelian-$n$-power} if it is non-empty, and $u_1 \approx_{a,k} u_2 \approx_{a,k} \ldots  \approx_{a,k} u_n$. Its \emph{period} is $\vert u_1\vert$.
Similarly, a \emph{$k$-abelian-square} (resp. \emph{$k$-abelian-cube}) is a $k$-abelian-2-power (resp. $k$-abelian-3-power). 
A word is said to be \emph{$k$-abelian-$n$-power-free} if none of its factors is a $k$-abelian-$n$-power.
Note that when $k=1$, the $k$-abelian-equivalence is exactly the abelian equivalence, and we ommit the ``1-'' prefix in this case.

The \emph{Parikh vector} of a word $w\in \Sigma^*$, denoted $\Psi(w)$, is the vector indexed by $\Sigma$ such that for every $a\in \Sigma$, $\Psi(w)[a]= |w|_a$. 
Then, by definition, two words $u$ and $v$ are abelian-equivalent if $\Psi(u)= \Psi(v)$.
For a set $S\subset \Sigma^*$ and a word $w\in \Sigma^*$, we denote by $\Psi_S(w)$ the vector indexed by $S$ such that for every $s\in S$, $\Psi_S(w)[s]= |w|_s$. 
We may write $\Psi_k(w)$ instead of $\Psi_{\Sigma^k}(w)$ if $\Sigma$ is clear in the context. %

For all $u\in \Sigma^*,\ i \leq |u|$, let $\operatorname{pref}_i(u)$ be the prefix of size $i$ of $u$ and $\operatorname{suf}_i(u)$ be the suffix of size $i$ of $u$. 
There are equivalent definitions of $k$-abelian equivalence (see \cite{Karhumaki1}). 
Two words of size at most $2k-1$ are $k$-abelian equivalent if and only if they are equal.
For every two words $u$ and $v$ of size at least $k-1$, the following conditions are equivalent: 
\begin{itemize}
	\item $u$ and $v$ are $k$-abelian equivalent (\emph{i.e.} $u \approx_{a,k} v$),
	\item $\Psi_k(u)= \Psi_k(v)$ and $\operatorname{pref}_{k-1}(u)=\operatorname{pref}_{k-1}(v)$,
	\item $\Psi_k(u)= \Psi_k(v)$ and $\operatorname{suf}_{k-1}(u)=\operatorname{suf}_{k-1}(v)$.
\end{itemize} 

\section{Abelian cubes and M\"akel\"a's Question~\ref{makquest}}\label{sec:mak}
Dekking showed that it is possible to avoid abelian-cubes in an infinite word over a ternary alphabet \cite{Ternarycube}. 
More recently Rao showed that one can avoid $2$-abelian-cubes over a binary alphabet \cite{abeliancube} and one can check that every word over a binary alphabet of length greater than 9 contains an abelian-cube.
The only open question about the avoidability of long $k$-abelian-cubes on infinite words is then the avoidability of long abelian-cubes over the binary alphabet.
This is the subject of Question \ref{makquest} from M\"akel\"a: he asked whether one can avoid every abelian-cubes of period at least 2 in binary words. We answer negatively this question.

For this,
we used a property of Lyndon words that made the exhaustive search much faster. 
A word $w\in\Sigma^*$ is a \emph{Lyndon word} if for all $u,v \in\Sigma^+$ such that $w=uv$, $w<_{lex}vu$, where $<_{lex}$ is the lexicographic order. 
The well known Chen-Fox-Lyndon Theorem 
states that every word may be written uniquely as a concatenation of non-increasing Lyndon words (see for example \cite{Lothaire}). In the following, we refer to this decomposition as the Lyndon factorization.
A language $L \subset \Sigma^*$ is \emph{factorial} if for every $w$ in $L$, every factor of $w$ is in $L$. 

\begin{Lemma}\label{Lyndonprop}
Any factorial language $L$ with arbitrarily long words contains arbitrarily long Lyndon words or repetitions of arbitrarily large power.
\end{Lemma}
\begin{proof}
Let us assume that there are no arbitrarily long Lyndon words in $L$.
This implies that there is a finite number $n$ of Lyndon words in $L$ and $s\in \mathbb{N}$ such that for every Lyndon word $w$ in $L$, $|w| \leq s$. 
Let $w_1, \ldots, w_n \in L^n$ be the Lyndon words of $L$ ordered by decreasing lexicographic order.

Then using the Lyndon factorization for every $w\in L$ there are some Lyndon words $L_1\geq_{lex}L_2\geq_{lex}\ldots \geq_{lex}L_d$ such that $w = L_1\ldots  L_d$. 
The fact that our language is factorial tells us that all the $L_i$ are in $L$. 
We get that for every $w\in L$ there are $\alpha_1, \ldots , \alpha_n\in \mathbb{N}$ such that $w= w_1^{\alpha_1}\ldots w_n^{\alpha_n}$. 
Then $|w| \leq  \sum_i |w_i| \times \alpha_i \leq s \times \sum_i \alpha_i \leq s \times n \times \operatorname{max_i}(\alpha_i).$

Since $L$ contains arbitrarily long words, then for any $t\in \mathbb{N}$ there is a $w\in L$ such that $|w|\geq t\times s \times n$.
Let $j\in\{1,\ldots,n\}$ be such that $\alpha_j=\operatorname{max_i}(\alpha_i)$. Thus $\alpha_j\geq t$ and then $(w_j)^t \in L$.
Thus we have arbitrarily long powers in $L$.
\end{proof}

A set of words that avoid certain kind of abelian repetitions is a factorial language and does not contain arbitrarily large powers.
Thus we just need to check that there are no arbitrarily long Lyndon words in it to deduce that this set does not contain arbitrarily long words.
The exhaustive search on prefixes of Lyndon words is then much shorter. %
Figure \ref{tree} shows how it helps for the exhaustive search of binary words avoiding abelian-squares of period at least 2.
The next Proposition answers negatively Question \ref{makquest}.

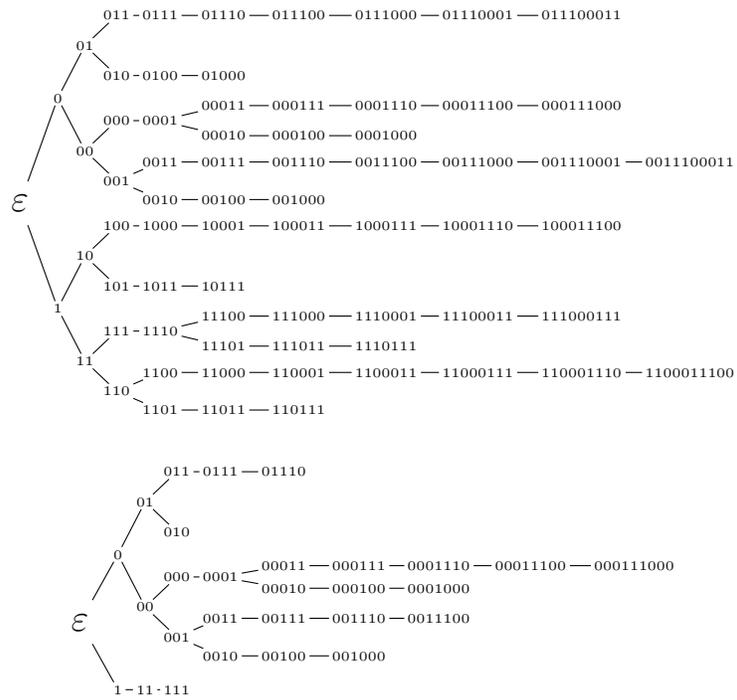
\begin{figure}
\centering
\begin{tikzpicture}[grow=right, sloped]
\tikzset{every node/.style={inner sep = 1pt, font=\tiny}}
\tikzstyle{level 1}=[anchor = west, sibling distance=28mm,level distance=4mm] 
\tikzstyle{level 2}=[sibling distance=14mm,level distance=2mm] 
\tikzstyle{level 3}=[sibling distance=8mm,level distance=2mm] 
\tikzstyle{level 4}=[sibling distance=5mm,level distance=3mm] 
\tikzstyle{level 5}=[sibling distance=4mm,level distance=5mm] 
\tikzstyle{level 6}=[sibling distance=4mm,level distance=6mm] 
\tikzstyle{level 7}=[sibling distance=3mm,level distance=7mm] 
\tikzstyle{level 8}=[sibling distance=2mm,level distance=7mm] 
\tikzstyle{level 9}=[sibling distance=2mm,level distance=8mm] 
\tikzstyle{level 10}=[sibling distance=2mm,level distance=8mm] 
\node[inner sep = 5pt] {\Large{$\varepsilon$} }
	child { node {1}
		child{ node{11}
			child{node{110}
				child{node{1101}	
					child{node{11011}	
						child{node{110111}}
					}
				}
				child{node{1100}
					child{node{11000}
						child{node{110001}
							child{node{1100011}
								child{node{11000111}
									child{node{110001110}
										child{node{1100011100}}
									}
								}
							}
						}
					}
				}
			}
			child{node{111}
				child{node{1110}
					child{node{11101}
						child{node{111011}
							child{node{1110111}}
						}
					}
					child{node{11100}
						child{node{111000}
							child{node{1110001}
								child{node{11100011}
									child{node{111000111}}
								}
							}
						}
					}
				}
			}
		}
		child{ node{10}
			child{ node{101}
				child{ node{1011}
					child{ node{10111}}
				}
			}
			child{ node{100}
				child{ node{1000}
					child{ node{10001}
						child{ node{100011}
							child{ node{1000111}
								child{ node{10001110}
									child{ node{100011100}}
								}
							}
						}
					}
				}
			}
		}
	}
child { node {0}
		child{ node{00}
			child{node{001}
				child{node{0010}	
					child{node{00100}	
						child{node{001000}}
					}
				}
				child{node{0011}
					child{node{00111}
						child{node{001110}
							child{node{0011100}
								child{node{00111000}
									child{node{001110001}
										child{node{0011100011}}
									}
								}
							}
						}
					}
				}
			}
			child{node{000}
				child{node{0001}
					child{node{00010}
						child{node{000100}
							child{node{0001000}}
						}
					}
					child{node{00011}
						child{node{000111}
							child{node{0001110}
								child{node{00011100}
									child{node{000111000}}
								}
							}
						}
					}
				}
			}
		}
		child{ node{01}
			child{ node{010}
				child{ node{0100}
					child{ node{01000}}
				}
			}
			child{ node{011}
				child{ node{0111}
					child{ node{01110}
						child{ node{011100}
							child{ node{0111000}
								child{ node{01110001}
									child{ node{011100011}}
								}
							}
						}
					}
				}
			}
		}
	}
;
\end{tikzpicture}

\vspace*{.6cm}

\centering
\begin{tikzpicture}[grow=right, sloped]
\tikzset{every node/.style={inner sep = 1pt, font=\tiny}}
\tikzstyle{level 1}=[anchor = west, sibling distance=18mm,level distance=4mm] 
\tikzstyle{level 2}=[sibling distance=14mm,level distance=2mm] 
\tikzstyle{level 3}=[sibling distance=8mm,level distance=2mm] 
\tikzstyle{level 4}=[sibling distance=5mm,level distance=3mm] 
\tikzstyle{level 5}=[sibling distance=3mm,level distance=5mm] 
\tikzstyle{level 6}=[sibling distance=2mm,level distance=6mm] 
\tikzstyle{level 7}=[sibling distance=2mm,level distance=6mm] 
\tikzstyle{level 8}=[sibling distance=2mm,level distance=7mm] 
\tikzstyle{level 9}=[sibling distance=2mm,level distance=8mm] 
\tikzstyle{level 10}=[sibling distance=2mm,level distance=8mm] 
\node[inner sep = 5pt] {\Large{$\varepsilon$} }
	child { node {1}
		child{ node{11}
			child{node{111}}
		}
	}
	child { node {0}
		child{ node{00}
			child{node{001}
				child{node{0010}	
					child{node{00100}	
						child{node{001000}}
					}
				}
				child{node{0011}
					child{node{00111}
						child{node{001110}
							child{node{0011100}}
						}
					}
				}
			}
			child{node{000}
				child{node{0001}
					child{node{00010}
						child{node{000100}
							child{node{0001000}}
						}
					}
					child{node{00011}
						child{node{000111}
							child{node{0001110}
								child{node{00011100}
									child{node{000111000}}
								}
							}
						}
					}
				}
			}
		}
		child{ node{01}
			child{ node{010}	}
			child{ node{011}
				child{ node{0111}
					child{ node{01110}
					}
				}
			}
		}
	}
;
\end{tikzpicture}
\caption{Top: the exhaustive search of binary words avoiding abelian-squares of period at least 2. Bottom: the same exhaustive search restricted on the prefixes of Lyndon words.}
\label{tree}
\end{figure}

\begin{Proposition}\label{nocube}
There is no infinite word over a binary alphabet avoiding abelian-cubes of period at least $2$. %
\end{Proposition}
We checked %
using a computer program that there are only finitely many Lyndon words over a binary alphabet avoiding abelian-cubes of period at least $2$. 
The program took approximately 3 hours to find all 2\,732\,711\,352 such Lyndon words and prefixes of Lyndon words. The longest word has a length of 290. 
Using Lemma \ref{Lyndonprop} we deduce that there is no infinite binary word avoiding abelian-cubes of size at least two.
Then we can reformulate the question and ask:
\begin{Problem}\label{iscube}
Is there a $p\in \mathbb{N}$ such that one can avoid abelian-cubes of period at least $p$ over two letters ? 
\end{Problem}
For $p=3$, we found a word of lenght 2\,500.

\section{Avoiding long $k$-abelian-squares on binary words}\label{sec:kab}
It is easy to verify that one cannot avoid squares of period at least 2 over a binary alphabet.
Entringer \emph{et al.} %
showed that it is possible to construct a binary word avoiding squares of period at least 3 \cite{Entringer1974159}. 
They also showed that every infinite binary word contains arbitrarily long abelian-squares.
Thus, one can wonder about the avoidability of large $k$-abelian-squares. Rao asked the following question:
\begin{Problem}[Rao \cite{abeliancube}]\label{kabeliansquarequestion}
What is the smallest k (if any) such that arbitrarily long $k$-abelian-squares can be avoided over a binary alphabet ?
\end{Problem}
Since arbitrarily long abelian squares cannot be avoided in binary words, $k$ is at least $2$.
In this section, we show that $k$ is at most $3$, that is one can avoid long $3$-abelian-squares over a binary alphabet, 
by giving a morphism whose fixed point avoids $3$-abelian-squares of period at least $3$.
A morphism $h$ is said \emph{$(p,k)$-abelian-square-free} if for every abelian-square-free word $w$, $h(w)$ avoids $k$-abelian-squares of period at least $p$. 
Let $h$ be the following morphism: %
$$h:  
\left\{
  \begin{array}{ll}
      0 \rightarrow & 00001101010\\ 
      1 \rightarrow & 00011111010\\ 
      2 \rightarrow & 00110100110\\ 
      3 \rightarrow & 00111001010.\\ 
    \end{array}
  \right.
$$
\begin{Theorem}\label{maintheo}
The morphism $h$ is (3,3)-abelian-square-free.
\end{Theorem}

\begin{proof}
\newmuskip{\medmuskipsave}
\medmuskipsave=\medmuskip
\setlength{\medmuskip}{0mu}

The proof is based on the same idea as the one used by Rao to give sufficient conditions for a morphism to be $k$-abelian-free \cite{abeliancube} which is a generalization of the sufficient conditions given by Carpi for abelian-free-morphisms \cite{carpi1}.
We checked the sufficient conditions on $h$ by computer.

In the proof we use the following property:
for every $k,i \in \mathbb{N}$ and $u,v\in \Sigma^*$ such that $i < k$, $k-1-i \leq |u| $ and $i \leq |v|$:
\begin{equation}
\Psi_k(uv)= \Psi_k(u \operatorname{pref}_i(v))+ \Psi_k(\operatorname{suf}_{k-1-i}(u)v) \label{split}\tag{S}.
\end{equation}

Let $A=\{0,1,2,3\}$ and $w \in A^*$. Let us show that if $h(w)$ contains a 3-abelian-square of period at least 3, then $w$ is not abelian-square-free.
Note that every image $h(x)$, $x\in A$, starts with the prefix $p=00$ and ends with the suffix $s=10$.

We check using a computer that for every $a,b \in A$, $ a\not=b$, every $3$-abelian-square of $h(ab)$ has period at most $2$. 
So if there is a forbidden 3-abelian-square, it has to be on the image of at least 3 letters. 
Then there are $a_1, a_2,a_3\in A$, $x_1,x_2 \in A^*$ and $(u_1, v_1), (u_2, v_2), (u_3, v_3) \in  (\{0,1\}^*, \{0,1\}^+)$  such that:
\begin{itemize}
\item $a_1x_1a_2x_2a_3$ is a factor of $w$,
\item for every $i \in \{1,2,3\}$, $u_iv_i =h(a_i)$,
\item $v_1h(x_1)u_2  \approx_{a,3} v_2h(x_2)u_3$.
\end{itemize}

Since $|v_1u_2v_2u_3|\ge 1+\vert h(a_2)\vert \geq 12$, either $|v_1u_2|\geq6$ or $|v_2u_3|\geq6$. 
Moreover, $|v_1h(x_1)u_2|= |v_2h(x_2)u_3|$, thus for all $i\in\{1,2\}$, $|v_ih(x_i)u_{i+1}| \geq 6$, and 
for all $i\in\{1,2\}$, $|v_i|\geq2$, $|h(x_i)|\geq2$ or $|u_{i+1}|\geq2$.
If $|u_{i+1}|\geq 2$, then:
\begin{align*}
\Psi_3(v_ih(x_i)u_{i+1}) \hspace{-2cm}\\
&=\Psi_3(v_i00) + \Psi_3(h(x_i)u_{i+1})\ \text{ (using  (\ref{split}) and $\operatorname{pref}_2(h(x_i)u_{i+1})=00$)}\\
&=  \Psi_3(v_i00) + \Psi_3(h(x_i)00) +  \Psi_3(u_{i+1})\ \text{ (using  (\ref{split}) and $\operatorname{pref}_2(u_{i+1})=00$)}\\
&= \Psi_3(v_i00) + \Psi_3(h(x_i)00) +  \Psi_3(10u_{i+1}) -  \Psi_3(1000).
\end{align*}
If $|v_{i}|\geq 2$ or $|h(x_i)|\geq 2$, we have the same result. So for every $i \in \{1,2\}$, we get:
\begin{equation}
 \Psi_3(v_ih(x_i)u_{i+1})= \Psi_3(v_i00) + \Psi_3(h(x_i)00) +  \Psi_3(10u_{i+1}) -  \Psi_3(1000) \label{formulaboutcut}\tag{L}.
\end{equation}

Let $N$ be the matrix indexed by $\{0,1\}^3 \times \{0,1,2,3\}$ with $N[w,x] = |h(x)00|_w$. %
$${}^t N =
\begin{pmatrix}
   3 & 1 & 2 & 1 & 1 & 2 & 1 & 0\\
   2 & 1 & 1 & 1 & 1 & 1 & 1 & 3\\
   1 & 2 & 1 & 2 & 2 & 1 & 2 & 0\\
   1 & 2 & 2 & 1 & 2 & 1 & 1 & 1\\
\end{pmatrix}
$$

For every word $w$, $\Psi_3( h(w)00)=N \Psi(w)$, thus the equality (\ref{formulaboutcut}) can be rewritten: 
$$\Psi_3(v_ih(x_i)u_{i+1})= \Psi_3(v_i00) + N \Psi(x_i) +  \Psi_3(10u_{i+1}) -  \Psi_3(1000).$$

Using $v_1h(x_1)u_2  \approx_{a,3} v_2h(x_2)u_3$, we get the following:
\begin{equation}
  N (\Psi(x_2)-\Psi(x_1)) =  \Psi_3(v_100) + \Psi_3(10u_2) - \Psi_3(v_200) - \Psi_3(10u_3).  \label{isinImN}
\end{equation}

Let $M$ be the sub-matrix of $N$ made of its rows 1, 2, 3 and 4 (they correspond to the words $000,001,010,011$).
$$M = \begin{pmatrix}
   3 & 2 & 1 & 1\\
   1 & 1 & 2 & 2\\
   2 & 1 & 1 & 2\\
   1 & 1 & 2 & 1\\
\end{pmatrix}$$
Let $S=\{000,001,010,011\}$.
Then $\Psi_S(w)$ is the sub-vector of $\Psi_3(w)$ made of the rows 1, 2, 3 and 4. We can check that $M$ is non-singular. Thus: %
\begin{equation}
\Psi(x_2) - \Psi(x_1) = M^{-1}(\Psi_S(v_100) +  \Psi_S(10u_2)-  \Psi_S(v_200)-\Psi_S(10u_3)).\label{Minisint}
\end{equation}

Let: $$\Psi_S(v_1, u_2, v_2, u_3) = \Psi_S(v_100) +  \Psi_S(10u_2) - \Psi_S(v_200)-\Psi_S(10u_3),\text{ and}$$
$$\Psi_3(v_1, u_2, v_2, u_3) = \Psi_3(v_100) +  \Psi_3(10u_2) - \Psi_3(v_200)-\Psi_3(10u_3).$$

\sloppy From equation (\ref{isinImN}), $ \Psi_3(v_1, u_2, v_2, u_3)$ is in $\operatorname{Im}(N)$,
and from equation (\ref{Minisint}), $M^{-1}(\Psi_S(v_1, u_2, v_2, u_3))$ is an integer vector.
From $v_1h(x_1)u_2  \approx_{a,3} v_2h(x_2)u_3$, if we note $p=00$, $s=10$ and $k=3$, we get $\operatorname{pref}_{k-1}(v_1p) = \operatorname{pref}_{k-1}(v_2p)$, $\operatorname{suf}_{k-1}(su_2) = \operatorname{suf}_{k-1}(su_3)$.
The following claim is verified using a computer program. There are $4^3$ values for the $a_i$ and $11^3$ ways of choosing the $u_i, v_i$ for each of them which makes 85\,184 cases to check (most of them are eliminated by the prefix and suffix conditions).
\begin{Claim}\label{centralclaim}
For all $a_1, a_2, a_3\in A$ and $(u_1, v_1), (u_2, v_2), (u_3, v_3) \in  \left(\{0,1\}^*, \{0,1\}^+\right)$ such that:
\begin{itemize}
  \item $ \forall i \in \{1,2,3\}$, $u_iv_i =h(a_i)$, 
  \item $\operatorname{pref}_{k-1}(v_1p)=\operatorname{pref}_{k-1}(v_2p)$ and $\operatorname{suf}_{k-1}(su_2)=\operatorname{suf}_{k-1}(su_3)$,
  \item $\Psi_{k}(v_1, u_2, v_2, u_3) \in \operatorname{Im}(N)$ and $M^{-1}(\Psi_S(v_1, u_2, v_2, u_3))$ is an integer vector,
\end{itemize}
there are $(\alpha_1, \alpha_2, \alpha_3)\in \{0,1\}$ such that:
\begin{equation}
M^{-1}(\Psi_S(v_1, u_2, v_2, u_3)) = \alpha_1 \Psi(a_1 )- (2\alpha_2 -1)\Psi(a_2)- (1-\alpha_3) \Psi(a_3). \label{eqclaim}\tag{E}
\end{equation}
\end{Claim}

From the claim we have $(\alpha_1, \alpha_2, \alpha_3)\in \{0,1\}$ such that equation (\ref{eqclaim}) is fulfilled.
Now we can introduce $x'_1, x'_2$ such that: $\forall  i \in  \{1,2\}$, $x'_i = a_i^{\alpha_i} x_i a_{i+1}^{1-\alpha_{i+1}}$.
Then $x'_1x'_2$ is a factor of $w$ and we have the following.
\begin{align*}
\Psi(x'_2)-\Psi (x'_1) &= \Psi(x_2)+ (1-\alpha_3) \Psi(a_3)+\alpha_2 \Psi(a_2) -(\Psi (x_1) +(1-\alpha_2)\Psi(a_2) + \alpha_1 \Psi(a_1) )\\
 &=M^{-1}(\Psi_S(v_1, u_2, v_2, u_3))- \alpha_1 \Psi(a_1 )+ (2\alpha_2 -1)\Psi(a_2)+ (1-\alpha_3) \Psi(a_3)\\
 \Psi(x'_2)-\Psi (x'_1) &=0
\end{align*}
This implies that there is an abelian-square on $w$.
\setlength{\medmuskip}{\medmuskipsave}
\end{proof}
Using Theorem~\ref{maintheo} together with the existence of abelian-square-free words over four letters, we get the following corollary.
\begin{Corollary}\label{col:33}
There is an infinite binary word that avoids 3-abelian-squares of period at least 3.
\end{Corollary}
Moreover, we can deduce the exponential growth of such words from the exponential growth of abelian-square-free words over four letters~\cite{carpi2}.
Corollary~\ref{col:33} gives a partial answer to Question \ref{kabeliansquarequestion}: there is such a $k$, and $k \in \{2,3\}$.  
We can then ask the following question. 
\begin{Problem}\label{question2abeliansquare}
Can we avoid $2$-abelian-squares of period at least $p$ on the binary alphabet,
for some $p\in \mathbb{N}$ ? 
\end{Problem}
Computer experiments show that we can avoid those patterns for $p=3$ in a word of length 15\,000.

\subsection*{2-abelian squares over a ternary alphabet}
Rao showed that one can build an infinite word that avoids 3-abelian-squares over a ternary alphabet \cite{abeliancube}. 
The longest 2-abelian-square-free ternary word has a length of 537 \cite{Karhumaki1}.
M\"akel\"a asked whether we can avoid abelian-squares of period at least 2 in ternary words (Question \ref{makquestsquares}). 
We give the answer to a weaker version of this question, that is one can avoid 2-abelian-squares of period at least 2 over the ternary alphabet.
Let:
$$h_2:  
\left\{
  \begin{array}{ll}
      0 \rightarrow & 00021\\ 
      1 \rightarrow & 00111\\ 
      2 \rightarrow & 01121\\ 
      3 \rightarrow & 01221.\\ 
    \end{array}
  \right.
$$
\begin{Theorem}\label{2abeliansquare}
$h_2$ is $(2,2)$-abelian-square-free.
\end{Theorem}
\begin{proof}
The proof is also done by checking sufficient conditions, similar to those in the proof of the Theorem \ref{maintheo}.
The Claim \ref{centralclaim} is true for this morphism with $k=2$, $S=\{00,01,02,11 \}$, $s=1$ and $p=0$.
\end{proof}

If $w$ is an infinite square-free-abelian word over four letters, $h_2(w)$ is a ternary word which avoid 2-abelian-squares of period at least 2.

\section{Minimal number of distinct 3-abelian-squares in infinite binary words}\label{sec:howmanykab}
Fraenkel and Simpson showed that there is an infinite binary word containing only the squares $0^2$, $1^2$, $(01)^2$ \cite{Fraenkel94howmany}. Moreover, every binary infinite word contains at least three distinct squares. 
It is natural to ask whether this property can be extended to the $k$-abelian case: is there a $k\in \mathbb{N}$ such that there is an infinite binary word that contains only 3 distinct $k$-abelian-squares?

More generally let $g(k)$ be the minimal number of distinct $k$-abelian squares that an infinite binary word must contain. 
Any $(k+1)$-abelian-square is a $k$-abelian-square so $g$ is non-increasing.
From Fraenkel and Simpson's result, we know that $g(k) \geq 3$ for all $k$.

\begin{table}
\footnotesize
$$h_3:
\left\{
  \begin{array}{ll}
    0 \rightarrow 
    &u1001011000101110001100101100010111001011000111001011100011001011000\\
    &10111001011001110001011000111001011100011001011000101110010110001110\\
    &01011100011001011000111001011001110001100101100011100101110001100101\\
    &10001011100101100011100101110001100101100011100101100111000101100011\\
    &10010110001011100101100111000101110010110001011100011001011000111001\\
    &0110011100010111001011000111001011100011v\\
    1 \rightarrow 
    &u0001110010110011100011001011000111001011001110001011100101100011100\\
    &10111000110010110001110010110011100011001011000111001011100010110001\\
    &11001011001110001100101100011100101100111000101100011100101100010111\\
    &00011001011000111001011001110001100101100011100101110001100101100010\\
    &11100101100011100101110001100101100011100101100111000110010110001110\\
    &0101110001100101100010111001011001110v\\
    2 \rightarrow
    &u0001110010110011100011001011000111001011001110001011100101100011100\\
    &10111000110010110001011100101100111000101100011100101100010111001011\\
    &00111000101110010110001110010111000110010110001011100101100111000101\\
    &11001011000101110001100101100010111001011001110001011000111001011001\\
    &11000110010110001110010111000110010110001011100101100111000101110010\\
    &110001011100011001011000101110010110011100011v\\
    3 \rightarrow
    &u0001110010110001011100011001011000101110010110001110010111000110010\\
    &11000111001011001110001100101100011100101110001011000111001011001110\\
    &00110010110001110010110011100010110001110010110001011100011001011000\\
    &10111001011001110001011000111001011100010110011100011001011000111001\\
    &01100111000101100011100101100010111001011001110001011100101100011100\\
    &101110001100101100010111001011001110v\\
    \end{array}
  \right.
$$
Where:\hfill~
$$
\begin{array}{ll}
u=&11000110010110001011100101100111000101100011100101100\\
  &01011100101100111000101110010110001011100011001011000\\
  &111001011001110001100101100010111001011001110001011\\
v=&00101100011100101100111000110010111000101100111000101\\
  &11001011000101110001100101110001011001110001100101100\\
  &01110010111000101100011100101100111000101100011100101\\
\end{array}
$$

\caption{A $(5,3)$-abelian-square-free morphism, with only three distinct $5$-abelian-squares: $00$, $11$ and $0101$.}
\label{tb:h3}
\end{table}

\begin{Proposition}
The morphism $h_3$ (defined in Table~\ref{tb:h3}) is $(5,3)$-abelian-square-free. Moreover, for every abelian-square-free word $w$, $h_3(w)$ contains only 3 distinct $5$-abelian-squares: $0^2$, $1^2$ and $(01)^2$.
\end{Proposition}
\begin{proof}
The proof that $h_3$ is $(5,3)$-abelian-square-free is also done by a computer check, similar to the proof of the Theorem \ref{maintheo}. 
This morphism is not uniform, so we need to check images of words of size up to 3 to ensure that
the image of a long 5-abelian square starts and ends on images of different letters.
\end{proof}
This proposition together with Ker\"anen's word tells us that for any $k \geq 5$ there is an infinite binary word with only 3 distinct $k$-abelian-squares, \emph{i.e.} $g(k) =3$ for every $k\ge 5$.
Propositions \ref{g34lower} and \ref{g34higher} give us that $g(3)=g(4)=4$.

\begin{Proposition}\label{g34lower}
Every word of size more than 87 over the binary alphabet contains at least 4 distinct 4-abelian-squares.
\end{Proposition}
This was verified by an exhaustive computer search. 

\begin{Proposition}\label{g34higher}
Let:
$$h_4: \left\{
  \begin{array}{ll}
      0 \rightarrow & 0001100101001101011000101010001011101011000101\\ 
      1 \rightarrow & 0001100101001101011001110101011100011101011000101\\ 
      2 \rightarrow & 0001100101001110001010001100101100011101011000101\\ 
      3 \rightarrow & 000110010100111001010100111000101100101011000101.\\ 
    \end{array}
  \right.
$$
Then $h_4$ is $(3,3)$-abelian-square-free. Moreover, for every abelian-square-free word $w$, $h_4(w)$ contains only 4 distinct $3$-abelian-squares: $0^2$, $1^2$, $(01)^2$ and $(10)^2$.
\end{Proposition}
The proof that $h_4$ is $(3,3)$-abelian-square-free is also done by a computer check, similar to the proof of the Theorem \ref{maintheo}. 
One can then check that $(00)^2$ and $(11)^2$ do not appear as factors of any image of a two-letter word. 

Finally using again an exhaustive search we were able to give the lower bound $g(2)\ge 5$.
\begin{Proposition}\label{g2lower}
Every word of size more than 92 over the binary alphabet contains at least 5 distinct 2-abelian-squares.
\end{Proposition}

The following question is a stronger version of Question \ref{question2abeliansquare}:
\begin{Problem}
How many distinct 2-abelian-squares must an infinite binary word contain?
\end{Problem}

\section{Computer searches}\label{sec:concl}

The toughest part to prove the existence of an infinite word with a desired property (say, the property P) is to find a morphism with sufficient conditions whose fixed point has the property P. 
Let $\mathcal{W}$ be the family of finite words with the property P.
This part was done by a computer-assisted search, as follows.
We look for a morphism whose images share a long common prefix and a long common suffix, to avoid the creation of small forbidden patterns when we concatenate two images.
For this, we first selected one ``good'' factor $w=uv$ (with $|u|=|v|$), where $v$ will be the common prefix, and $u$ the common suffix. To optimize the chance of success, we took a factor $w$ which is supposed to appear often in words in $\mathcal{W}$:
we constructed, with a backtrack algorithm, a long random word in $\mathcal{W}$, and selected the factor $w$ among the factors which appear most often is this word.
Then we constructed a family $\mathcal{F}_w$ of words in $\mathcal{W}$ with prefix $u$, suffix $v$, and with a size of at most a fixed number. 
(If the search fails, we try again with another $w$, or with a larger size.)
We constructed the graph with vertex set $\mathcal{F}_w$, and edges $\{x,y\}$ such that $xyx$ and $yxy$ are in $\mathcal{W}$.
Finally, we checked the sufficient conditions on every morphism which corresponds to a clique of size $4$ in the graph. %

As an example, for the morphism $h_3$, which was the most arduous to find among the morphisms presented here, we use the following parameters.
We selected the $88$ factors with a size of $300$ which appear most often among approximately 1.2 million factors we found in the random word.
We computed the families $\mathcal{F}_w$ with words of size at most $700$.
All the families $\mathcal{F}_w$ we computed had a size between 1\,000 and 5\,000.
The graph of the family that gave us the morphism $h_3$ had 58\,680 edges and 1\,977 vertices (density of $\sim$0.03).
With the right parameters ($|w|$ and length of the elements of $\mathcal{F}_w$) it took half a day to find a good morphism. 

The presented approach cannot be directly used to answer positively to Question~\ref{question2abeliansquare}. For every binary word $w$, $\vert wx\vert_{10}=\vert wx\vert_{01}$ (where $x$ is the first letter of $w$). So the matrix $N$ has rank at most $3$, and one cannot find an invertible sub-matrix $M$ of size $4$. 
Thus Question~\ref{question2abeliansquare} is, in spirit, close to M\"akel\"a's questions (Question~\ref{makquestsquares} and Question~\ref{iscube}, the modified version of Question~\ref{makquest}).

\bibliographystyle{plain}
\bibliography{biblio}

\begin{thebibliography}{10}

\bibitem{carpi1}
A.~Carpi.
\newblock On abelian power-free morphisms.
\newblock {\em International Journal of Algebra and Computation},
  03(02):151--167, 1993.

\bibitem{carpi2}
A.~Carpi.
\newblock On the number of abelian square-free words on four letters.
\newblock {\em Disc. Appl. Math}, 81(1-3):155--167, 1998.

\bibitem{Ternarycube}
F.~M. Dekking.
\newblock Strongly non-repetitive sequences and progression-free sets.
\newblock {\em Journal of Combinatorial Theory, Series A}, 27(2):181 -- 185,
  1979.

\bibitem{Entringer1974159}
R.~C. Entringer, D.~E. Jackson, and J.~A. Schatz.
\newblock On nonrepetitive sequences.
\newblock {\em Journal of Combinatorial Theory, Series A}, 16(2):159 -- 164,
  1974.

\bibitem{erdos1}
P.~Erd{\H{o}}s.
\newblock Some unsolved problems.
\newblock {\em The Michigan Mathematical Journal}, 4(3):291--300, 1957.

\bibitem{erdos2}
P.~Erd{\H{o}}s.
\newblock Some unsolved problems.
\newblock {\em Magyar Tud. Akad. Mat. Kutat\'o Int. K\"ozl.}, 6:221--254, 1961.

\bibitem{Evdokimov1}
A.~A. Evdokimov.
\newblock Strongly asymmetric sequences generated by a finite number of
  symbols.
\newblock {\em Dokl. Akad. Nauk SSSR}, 179:1268--1271, 1968.

\bibitem{Fraenkel94howmany}
A.~S. Fraenkel and R.~J. Simpson.
\newblock How many squares must a binary sequence contain?
\newblock {\em The Electronic Journal of Combinatorics}, 2, 1995.

\bibitem{Karhumaki1}
J.~{Karhumaki}, A.~{Saarela}, and L.~Q. {Zamboni}.
\newblock {On a generalization of Abelian equivalence and complexity of
  infinite words}.
\newblock {\em Journal of Combinatorial Theory, Series A}, 120(8):2189--2206,
  2013.

\bibitem{keranen1}
V.~Ker{\"a}nen.
\newblock Abelian squares are avoidable on 4 letters.
\newblock In {\em ICALP}, pages 41--52, 1992.

\bibitem{keranen2}
V.~Ker{\"a}nen.
\newblock New abelian square-free {DT0L}-languages over 4 letters.
\newblock {\em Manuscript}, 2003.

\bibitem{Lothaire}
M.~Lothaire.
\newblock {\em Combinatorics on Words}.
\newblock Cambridge University Press, 1997.

\bibitem{pleasant1}
P.~A.~B. Pleasants.
\newblock Non-repetitive sequences.
\newblock {\em Mathematical Proceedings of the Cambridge Philosophical
  Society}, 68:267--274, 9 1970.

\bibitem{abeliancube}
M.~Rao.
\newblock On some generalizations of abelian power avoidability.
\newblock {\em Manuscript}, 2014.

\bibitem{Thue1}
A.~Thue.
\newblock {\"Uber} die gegenseitige {L}age gleicher {T}eile gewisser
  {Z}eichenreihen.
\newblock {\em Norske Vid. Selsk. Skr. I. Mat. Nat. Kl. Christiania}, 10:1--67,
  1912.

\end{thebibliography}

\end{document}